\documentclass[letterpaper, 10 pt, conference]{ieeeconf}  

\IEEEoverridecommandlockouts                              

\usepackage{amsfonts}
\usepackage{amsmath}
\usepackage{graphicx}
\usepackage{xcolor}
\usepackage{cite}

\newtheorem{definition}{Definition}
\newtheorem{proposition}{Proposition}
\newtheorem{remark}{Remark}
\newtheorem{theorem}{Theorem}

\title{\LARGE \bf
Stability/instability study of density systems and control law design
}

\author{Igor Furtat$^{1}$
\thanks{$^{1}$Igor Furtat is with IPME RAS, St. Petersburg, Bolshoj pr. V.O., 61, 199178, Russia
        {\tt\small cainenash@gmail.com}}%
}

\begin{document}

\maketitle
\thispagestyle{empty}
\pagestyle{empty}

\begin{abstract}

The paper considers some class of dynamical systems that called density systems. 
For such systems the derivative of quadratic function depends on so-called density function.
The density function is used to set the properties of phase space, therefore, it influences the behaviour of investigated systems. 
A particular class of such systems is previously considered for (in)stability study of dynamical systems using the flow and divergence of a phase vector.
In this paper, a more general class of such systems is considered, and it is shown that the density function can be used not only to study (in)stability, but also to set the properties of space in order to change the behaviour of dynamical systems.
The development of control laws based on use the density function for systems with known and unknown parameters is considered.
All obtained results are accompanied by the simulations illustrating the theoretical conclusions.

\end{abstract}

\section{INTRODUCTION}

The paper considers a class of dynamical systems in normal form,
where the right-hand side depends on some function that sets the properties of phase space and affects the behaviour of these systems.
This function is called the \textit{density function}.
All relevant definitions will be considered in the next section of the paper.

A particular class of such systems is considered in  
\cite{Zaremba54,Krasnoselski63,Fronteau65,Brauchli68,Zhukov78,Shestakov78,Jukov90,Jukov99,Rantzer01,Furtat20a,Furtat20b,Furtat21,Furtat22}, 
where a new system
$\dot{x}=\rho(x)f(x)$ is introduced for stability study of initial one $\dot{x}=f(x)$, where $\rho(x)>0$ is some auxiliary function.
Then the problems of (in)stability of such systems are studied using the properties of divergence and flow of the phase vector.
In \cite{Rantzer01} the function $\rho(x)$ is first called the density function,
and in \cite{Furtat20a,Furtat20b,Furtat21,Furtat22} it is shown the relationship between the obtained (in)stability criteria and the continuity equation,
which describes various processes in electromagnetism \cite{Griffiths17}, wave theory \cite{Arnold14}, hydrodynamics \cite{Pedlosky79}, mechanics of deformable solids \cite{Arnold14}, and quantum mechanics \cite{McMahon13}.

The papers \cite{Liberzon13,Bechlioulis14,Berger18,Furtat21b,Furtat21c} propose a number of control laws,
guaranteeing the presence of outputs in the sets specified by the designer.
To achieve such goal a some auxiliary function is introduced to the right-hand side of the original or transformed system to get
given properties in the closed-loop system.
Thus, in \cite{Liberzon13,Berger18} the control law with a funnel effect is proposed.
In \cite{Bechlioulis14} the prescribed performance control law is obtained,
which guarantees the finding of transients in a tube converging to a neighborhood of zero.
Differently from \cite{Liberzon13,Bechlioulis14,Berger18}, in \cite{Furtat21b,Furtat21c} new control laws allow to guarantee the location of outputs in the set that may be asymmetric with respect to the equilibrium position and does not converge to a given constant.

Differently from existing results, we consider a new class of dynamical systems that explicitly or implicitly depend on the density function.
By using this function, the density of space can be changed in the sense of selection of 
(in)stability regions, restricted regions (where there are no system solutions), and the regions with different values of the density that affects
on the behaviour of the investigated systems.

Thus, the contribution of the present paper is as follows:
\begin{enumerate}
\item[(i)] differently from \cite{Zaremba54,Krasnoselski63,Fronteau65,Brauchli68,Zhukov78,Shestakov78,Jukov90,Jukov99,Rantzer01,Furtat20a,Furtat20b,Furtat21,Furtat22} we consider systems where the density function is not necessarily multiplied by its entire right side;
\item [(ii)] in contrast to \cite{Liberzon13,Bechlioulis14,Berger18,Furtat21b,Furtat21c}, the density function may be presented implicitly on the right-hand side of the system;
\item [(iii)] differently from \cite{Liberzon13,Bechlioulis14,Berger18,Furtat21b,Furtat21c}, the density function can guarantee solutions in an unbounded set or the set with piecewise continuous in time boundaries;
\item [(iv)] in contrast to \cite{Furtat21b,Furtat21c}, the development of control laws does not require coordinate transformations, which makes it possible to consider the systems with unknown parameters.
\end{enumerate}

The paper is organized as follows.
Section \ref{Sec2} gives motivating examples, definitions of the density function and density systems.
Some properties of these systems are proved.
Section \ref{Sec3} proposes the application of the obtained results to design control laws for systems with known and unknown parameters.
The simulations of the obtained control schemes illustrate the confirmation of theoretical conclusions.

The following \textit{notations} are used in this work:
$\mathbb R^{n}$ is an $n$-dimensional Euclidean space with norm $|\cdot|$;
$\mathbb R_+$ ($\mathbb R_-$) is a set of positive (negative) real numbers;
$\mathbb R^{n \times m}$ is a set of all $n \times m$ real matrices.

\section{\uppercase{Motivating examples. Definitions}}
\label{Sec2}

Before introducing definitions, consider two examples.

\textit{Example 1.}
It is well known that the solutions of the system
\begin{equation}
\label{eq_ex1_00}
\begin{array}{l} 
\dot{x}=-x
\end{array}
\end{equation} 
converge asymptotically to zero, where $x \in \mathbb R$.
Multiplying the right-hand side of \eqref{eq_ex1_00} by a piecewise continuous in $t$ and locally Lipschitz in $x$ function $\rho(x,t): \mathbb R \times [0, +\infty) \to \mathbb R$, rewrite \eqref {eq_ex1_00} as follows
\begin{equation}
\label{eq_ex1_0}
\begin{array}{l}
\dot{x}=-\rho(x,t) x.
\end{array}
\end{equation}
A behaviour of the system \eqref{eq_ex1_0} depends on the properties of $\rho(x,t)$.
Using the function $\rho(x,t)$, one can set some properties and restrictions in the space $(x,t)$.
Thus, it is possible to influence the quality of transients of the original system \eqref{eq_ex1_00} and change its qualitatively.
Therefore, $\rho(x,t)$ is called the \textit{density function}.
Let us consider several examples how this function influences on behaviour of the system \eqref{eq_ex1_0}.

1) The density function $\rho(x,t)=\alpha >0$ holds the equilibrium $x=0$ and takes the same positive value for any $x$ and $t$, thus, it does not qualitatively affect the exponential stability of the trajectories of the original system \eqref{eq_ex1_00} (see Fig. \ref{Fig1}, \textit{a}). 
The value of $\rho(x,t)$ is influenced on the rate of convergence of the solution \eqref{eq_ex1_0} to the equilibrium point depending on the value of $\alpha$.
Indeed, choosing quadratic function (it can be Lyapunov function) $V=0.5x^2$, we obtain $\dot{V}=-\alpha x^2 < 0$ in the region $D_S =\mathbb R \setminus \{0\}$.

2) The density function $\rho(x,t)=\frac{\alpha}{w(t)-|x(t)|}$ with a piecewise continuous function $w(t)>0$ preserves a unique equilibrium point $x=0$ and takes positive values in the region $D_S=\{x \in \mathbb R: -w<x<w\}$, while for $|x| \to w-0$ we have $\rho(x,t) \to +\infty$.
These properties guarantee the asymptotic stability of the solutions \eqref{eq_ex1_0} in the region $D_S$, as well as the trajectories of \eqref{eq_ex1_0} never leave this region (see Fig. \ref{Fig1}, \textit{b}).
Choosing the quadratic function $V=0.5x^2$, we obtain $\dot{V}=-\frac{\alpha}{w-|x|} x^2 < 0$ in the domain $x \in D_S \setminus \{0\}$, that confirms the conclusion drawn.

\begin{figure}[h!]
\begin{center}
\begin{minipage}[h]{0.8\linewidth}
\center{\includegraphics[width=0.6\linewidth]{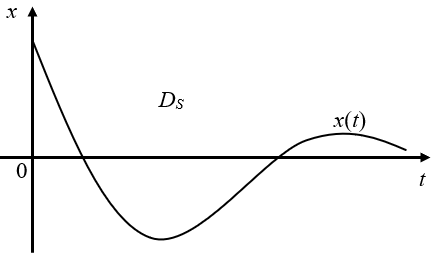}} \\ \textit{a}
\end{minipage}
\vfill
\begin{minipage}[h]{0.8\linewidth}
\center{\includegraphics[width=0.6\linewidth]{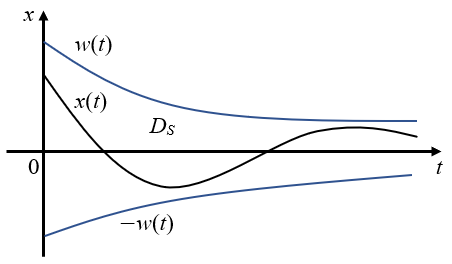}} \\ \textit{b}
\end{minipage}
\caption{The transients in the system \eqref{eq_ex1_0} with the density functions $\rho(x,t)=\alpha$ (\textit{a}) and $\rho(x,t)=\frac{\alpha}{w(t)-| x(t)|}$ (\textit{b}).}
\label{Fig1}
\end{center}
\end{figure}

3) Considering the density function $\rho(x,t)=\alpha [x(t)-w(t)] \textup{sign}(x)$ with a piecewise continuous function $w(t)$, we have two equilibriums $x=0$ and $x=w$.
The function $\rho(x,t)$ takes a positive value in the region $D_S=\{x \in \mathbb R: x \in (-\infty;0) \cup (w;+\infty)\}$ and negative value in $D_U=\{x \in \mathbb R: x \in (0;w)\}$.
In $D_S$ the system \eqref{eq_ex1_0} is stable, while it is unstable in $D_U$, which guarantees that $x(t)$ follows $w(t)$ (see Fig. \ref{Fig2}, \textit{a}) .
Choosing the quadratic function $V=0.5x^2$, we get $\dot{V}=-\alpha [x-w] \textup{sgn} (x) x^2 <0$ for $x \in D_S$ and $\dot{V}>0$ for $x \in D_U$.

4) Considering the density function $\rho(x,t)= -\alpha \ln \frac{\overline{w}(t)-x(t)}{x(t)-\underline{w}(t) }$ with piecewise continuous functions $w(t)>0$, $\overline{w}(t)>\underline{w}(t)>0$ we have one equilibrium $x(t)=w(t)=0.5[\overline {w}(t)+\underline{w}(t)]$.
The function $\rho(x,t)$ takes positive values in the region $D_S=\{x \in \mathbb R: w < x < \overline{w} \}$ and negative values in the region $D_U=\{x \in \mathbb R: \underline{w}<x <w \}$.
Hence, the system \eqref{eq_ex1_0} is stable in $D_S$, and it is unstable in $D_U$. 
Therefore, $x(t)$ follows $w(t)$, while in the shaded region the system \eqref{eq_ex1_0} has no solutions (see Fig. \ref{Fig2}, \textit{b}).
Moreover, the trajectories of \eqref{eq_ex1_0} never leave the region $D_S \cup D_U$ because $|\rho(x,t)| \to +\infty$ as $x$ approaches $\underline{w}$ and $\overline{w}$.
Choosing the quadratic function $V=0.5x^2$, we obtain $\dot{V}=\alpha \ln \frac{\overline{w}-x}{x-\underline{w}} x^2<0 $ for $x \in D_S$ and $\dot{V}>0$ for $x \in D_U$.

\begin{figure}[h!]
\begin{center}
\begin{minipage}[h]{0.8\linewidth}
\center{\includegraphics[width=0.6\linewidth]{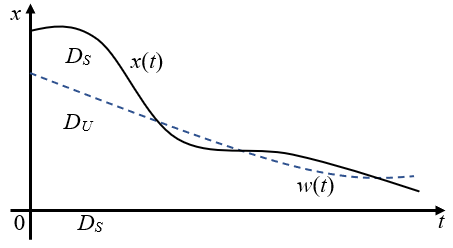}} \\ \textit{a}
\end{minipage}
\vfill
\begin{minipage}[h]{0.8\linewidth}
\center{\includegraphics[width=0.6\linewidth]{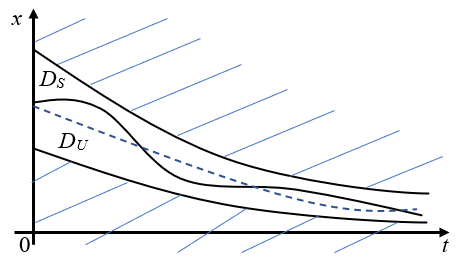}} \\ \textit{b}
\end{minipage}
\caption{The transients in the system \eqref{eq_ex1_0} with density functions $\rho(x,t)=\alpha [x(t)-w(t)] \textup{sign}(x)$ (\textit{a}) and $\rho (x,t)=-\alpha \ln \frac{\overline{w}(t)-x(t)}{x(t)-\underline{w}(t)}$ (\textit{b}).}
\label{Fig2}
\end{center}
\end{figure}

5) Considering the density function $\rho(x,t)= \alpha \ln (x(t)-g(t))$ with a piecewise continuous function $g(t)>0$, we have one equilibrium $ x(t)=w(t)=1+g(t)$.
The function $\rho(x,t)$ takes positive values in the region $D_S=\{x \in \mathbb R_+: w< x < +\infty\}$ and negative values in the region $D_U=\{x \in \mathbb R_+: g<x<w\}$.
Hence, the system \eqref{eq_ex1_0} is stable in $D_S$ and unstable in $D_U$, as well as there are no solutions of the system in the shaded region, which guarantees that $x(t)$ follows the trajectory $w(t)$ and slides along the shaded region (see Fig. \ref{Fig3}).
In this case, the trajectories of the system never enter the shaded region, because $\rho(x,t) \to -\infty$ when trajectory approaches the boundary $g(t)$.
Choosing the quadratic function $V=0.5x^2$, we get $\dot{V}=-\alpha \ln (x-g) x^2<0$ for $x \in D_S$ and $\dot{V}>0 $ for $x \in D_U$.

\begin{figure}[h]
\center{\includegraphics[width=0.5\linewidth]{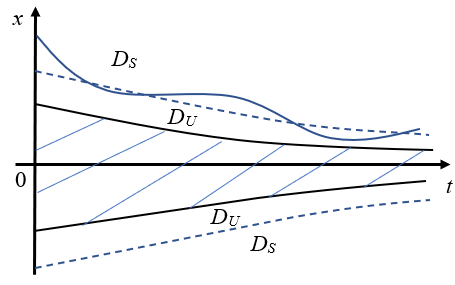}}
\caption{The transients in the system \eqref{eq_ex1_0} with the density function $\rho(x,t)=\alpha \ln (x(t)-g(t))$.}
\label{Fig3}
\end{figure}

Summarizing, Example 1 shows how the density function $\rho(x,t)$ given in the space $(x,t)$ can qualitatively influence the transients of the original system \eqref{eq_ex1_00}.
Unlike Example 1 in the following one, we will show that the density function can be explicitly or implicitly presented on the right-hand side of the system.


\textit{Example 2.}
Consider the system
\begin{equation}
\label{eq_ex2_1}
\begin{array}{l}
\dot{x}_1=x_2-\rho_1(x,t) x_1,
\\
\dot{x}_2=-x_1-\rho_2(x,t) x_2,
\end{array}
\end{equation}
where $\rho_1(x,t)$ and $\rho_2(x,t)$ are piecewise continuous functions in $t$ and locally Lipschitz in $x$ on $\mathbb R^2 \times [0, + \infty)$.
Consider the quadratic function
\begin{equation}
\label{eq_ex2_2}
\begin{array}{l}
V=0.5(x_1^2+x_2^2).
\end{array}
\end{equation}
Taking the derivative of \eqref{eq_ex2_2} along the solutions \eqref{eq_ex2_1}, we get
\begin{equation}
\label{eq_ex2_3}
\begin{array}{l}
\dot{V} = - \rho_1 x_1^2 - \rho_2 x_2^2.
\end{array}
\end{equation}

1) Let $\rho_1=\rho_2=-\ln \frac{\overline{g}-|x_1|^{\beta}-|x_2|^{\beta}}{|x_1|^{\beta}+| x_2|^{\beta}-\underline{g}}$,
where $\beta>0$ and $\overline{g}(t)>\underline{g}(t)>0$ are piecewise continuous functions.
Then $\dot{V}=- \rho(x,t)(x_1^2+x_2^2)$,
where $\rho(x,t)=-\ln \frac{\overline{g}-|x_1|^{\beta}-|x_2|^{\beta}}{|x_1|^{\beta}+| x_2|^{\beta}-\underline{g}}>0$ in
$D_S=\{x \in \mathbb R^2: \underline{g}<|x_1|^{\beta}+|x_2|^{\beta}<\overline{g}\}$ and
$\rho(x,t)<0$ in $D_U=\{x \in \mathbb R^2: \overline{g}<|x_1|^{\beta}+|x_2|^{\beta}< \underline{g}\}$.
In this case, the density function $\rho(x,t)$ is explicitly presented in the system \eqref{eq_ex2_1},
but it is not multiplied by the entire right side, unlike in example 1.
Fig.~\ref{Fig4} shows simulation results for $\beta=1$, $\overline{g}=3$, $\underline{ g}=2$ (\textit{a}) and $\beta=0.6$, $\overline{g}=3^{0.6}$, $\underline{g}=1$ (\textit{b}) with $x(0)=col\{0,2.5\}$.

Here and below:
\begin{itemize}
\item gray areas in the figures mean that the density function is chosen such that there are no solutions of the system in these areas (the density value increases to infinity on the boundary of this area);
\item the dotted curve indicates the equilibrium position and, accordingly, the boundary between the stable $D_S$ and unstable $D_U$ regions.
\end{itemize}

2) Let $\rho_1=1-x_1^2$ and $\rho_2=-x_1^2$.
Then $\dot{V}=- \rho(x,t)x_1^2$,
where $\rho(x,t)=x_1^2+x_2^2-1>0$ in $D_S=\{x \in \mathbb R^2: x_1^2+x_2^2>1\}$ and
$\rho(x,t)<0$ in $D_U=\{x \in \mathbb R^2: x_1^2+x_2^2<1\}$.
In this case, the density function $\rho(x,t)$ is presented implicitly in \eqref{eq_ex2_1},
unlike the previous case.
Fig.~\ref{Fig5} shows the simulation results for $x(0)=col\{2,1\}$.

3) Let $\rho_1=\alpha \ln(|x_1|^{\beta}+|x_2|^{\beta}-1)$, $\alpha>0$, $\beta>0$ and $\rho_2=0$.
Then $\dot{V}=- \rho(x,t)x_1^2$,
where $\rho(x,t)=\alpha \ln(|x_1|^{\beta}+|x_2|^{\beta}-1)>0$ in $D_S=\{x \in \mathbb R ^2:|x_1|^{\beta}+|x_2|^{\beta}>1\}$ and
$\rho(x,t)<0$ in $D_U=\{x \in \mathbb R^2: |x_1|^{\beta}+|x_2|^{\beta}<1\}$.
Unlike cases 1 and 2 now the density function $\rho(x,t)$ is presented in only one of the equations \eqref{eq_ex2_1}.
Fig.~\ref{Fig6} shows simulation results for $\beta=1$ (\textit{a}) and $\beta= 0.5$ (\textit{b}), as well as $\alpha=20$ and $x(0)=col\{2,2\}$.

\begin{figure}[h]
\begin{center}
\begin{minipage}[h]{0.8\linewidth}
\center{\includegraphics[width=0.7\linewidth]{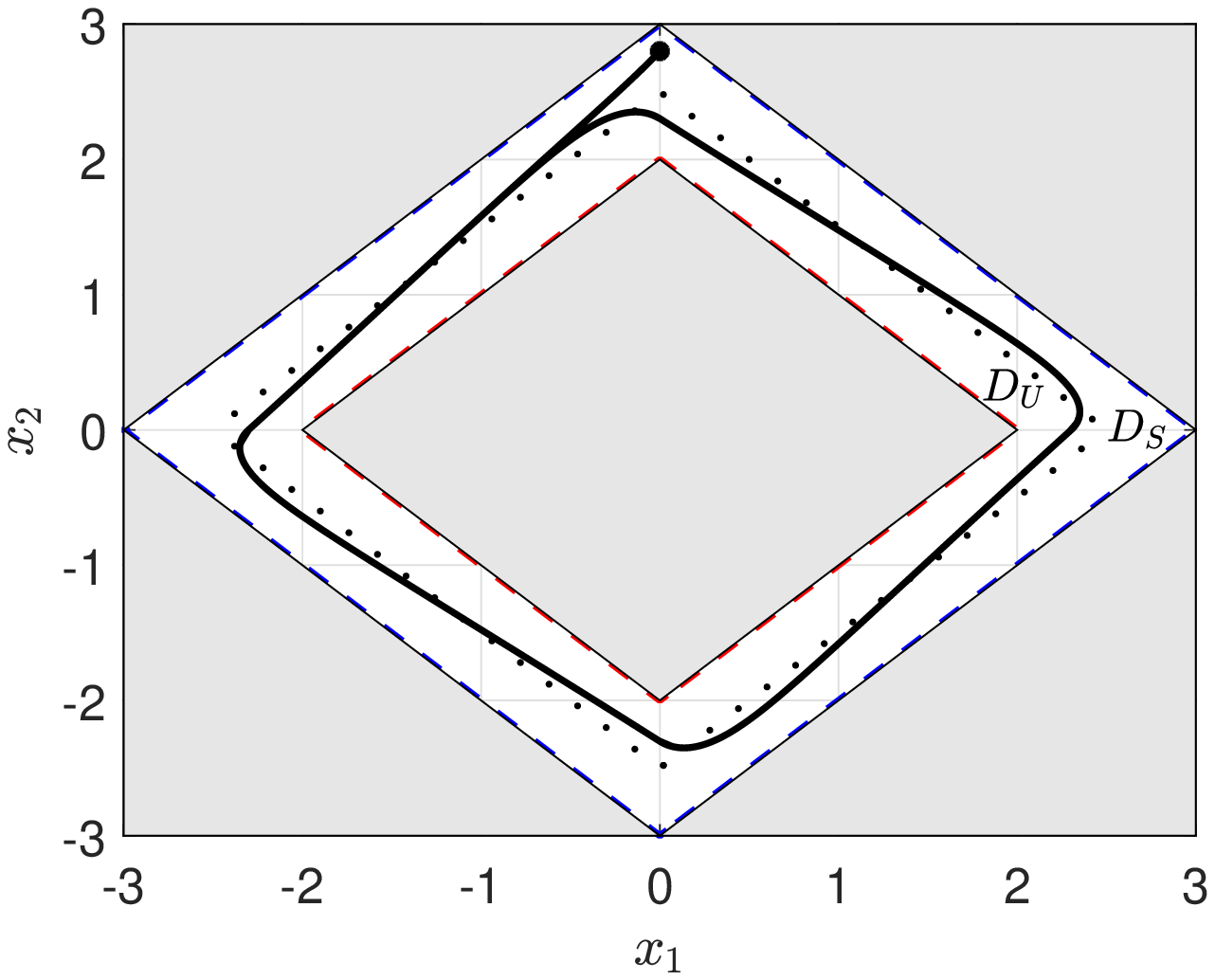}} \\ \textit{a}
\end{minipage}
\vfill
\begin{minipage}[h]{0.9\linewidth}
\center{\includegraphics[width=0.6\linewidth]{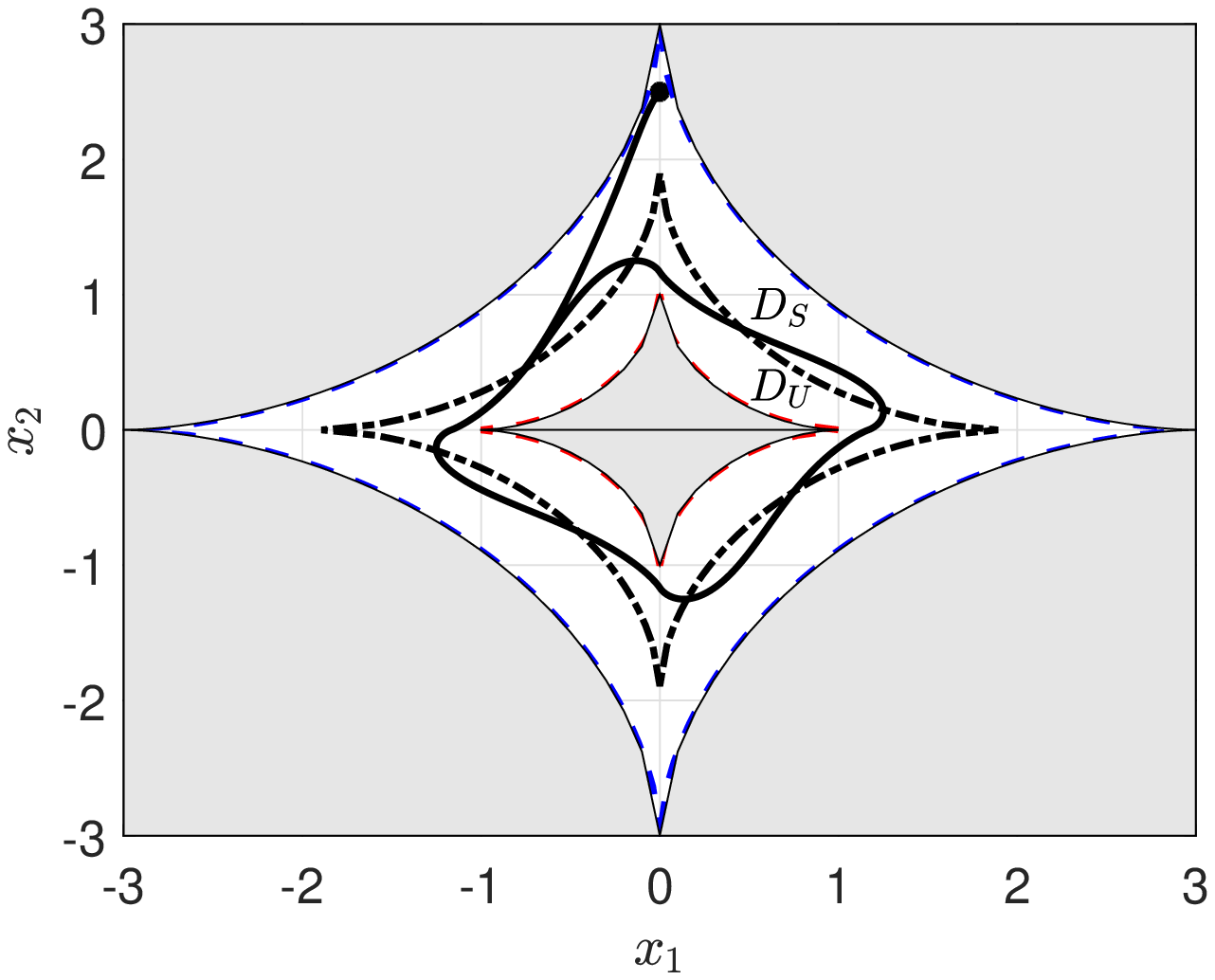}} \\ \textit{b}
\end{minipage}
\caption{The phase trajectories of the system \eqref{eq_ex2_1} with density functions $\rho(x,t)=-\ln \frac{3-|x_1|-|x_2|}{|x_1|+|x_2|-2}$ (\textit{a}) and $\rho(x,t)=-\ln \frac{3^{0.6}-|x_1|^{0.6}-|x_2|^{0.6}}{|x_1|^{0.6}+| x_2|^{0.6}-1}$ (\textit{b}).}
\label{Fig4}
\end{center}
\end{figure}

\begin{figure}[h]
\center{\includegraphics[width=0.6\linewidth]{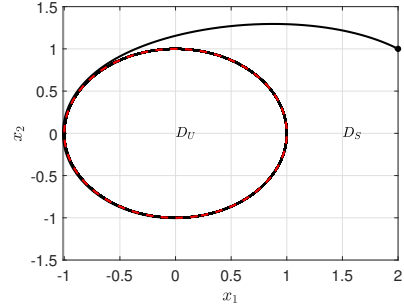}}
\caption{The phase trajectory of the system \eqref{eq_ex2_1} with density function $\rho(x,t)=x_1^2+x_2^2-1$.}
\label{Fig5}
\end{figure}

\begin{figure}[h]
\begin{center}
\begin{minipage}[h]{0.8\linewidth}
\center{\includegraphics[width=0.7\linewidth]{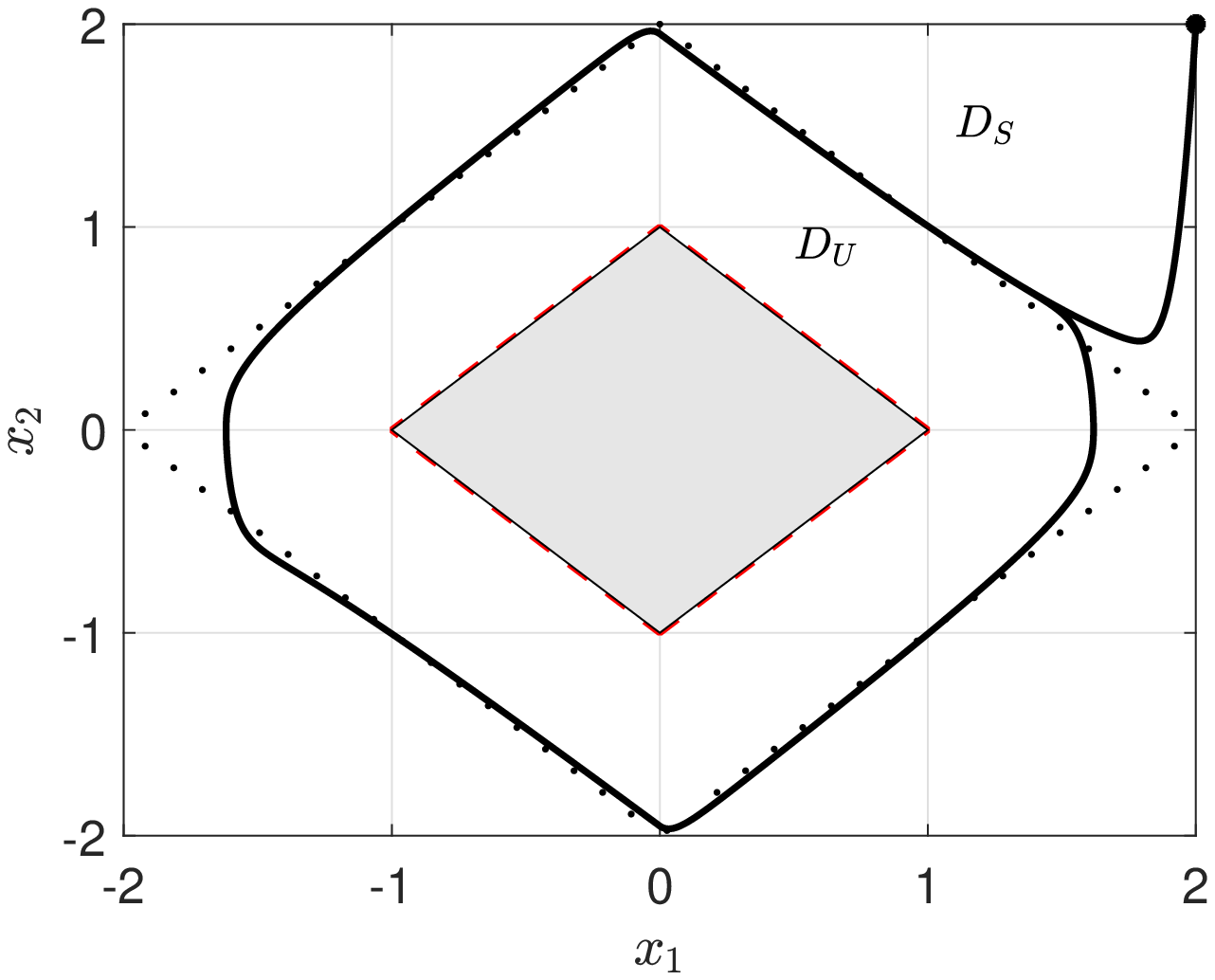}} \\ \textit{a}
\end{minipage}
\hfill
\begin{minipage}[h]{0.8\linewidth}
\center{\includegraphics[width=0.7\linewidth]{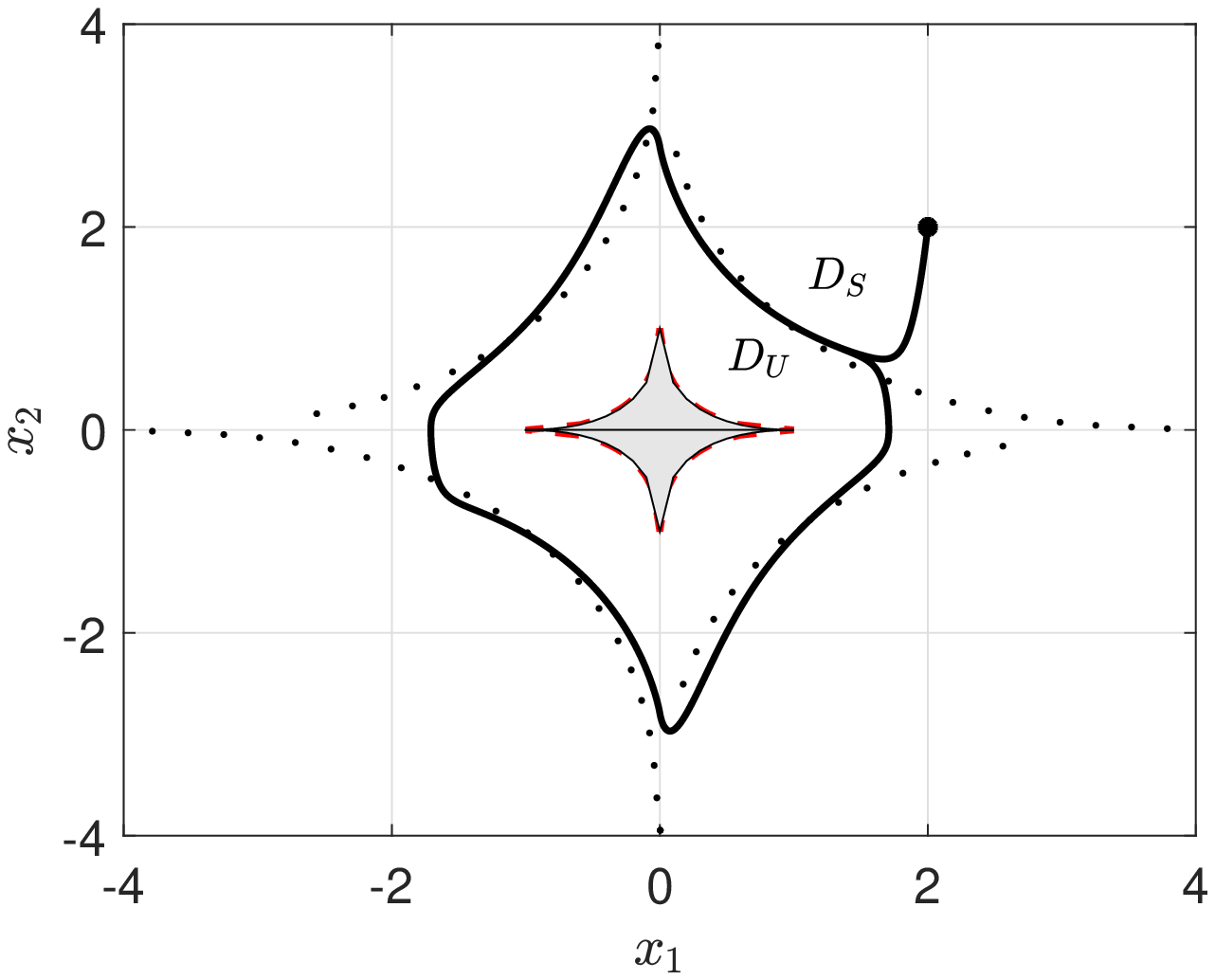}} \\ \textit{b}
\end{minipage}
\caption{The phase trajectories of the system \eqref{eq_ex2_1} with density functions $\rho(x,t)=20 \ln(|x_1|+|x_2|-1)$ (\textit{a}) and $\rho(x,t)=20 \ ln(|x_1|^{0.5}+|x_2|^{0.5}-1)$ (\textit{b}).}
\label{Fig6}
\end{center}
\end{figure}


Now we consider a dynamical system of the form
\begin{equation}
\label{eq1}
\begin{array}{l}
\dot{x}=f(x,t),
\end{array}
\end{equation}
where $x \in D \subset \mathbb R^n$ is the state,
the function $f: D \times [0, +\infty) \to \mathbb R^{n}$ is piecewise continuous in $t$ and locally Lipschitz in $x$ on $D \times [0, +\infty)$.

\begin{definition}
\label{def1}
The system \eqref{eq1} is called density with the density function $\rho(x,t): D \times [0, +\infty) \to \mathbb R$,
if there exists a continuously differentiable function $V(x,t): D \times [0, +\infty) \to \mathbb R$ such that
\begin{itemize}
\item[(a)] $w_1(x) \leq V(x,t) \leq w_2(x)$,
\item[(b)] $\dot{V} \leq \rho(x,t) W_1(x) \leq 0$ or $\dot{V} \geq \rho(x,t) W_2(x) \geq 0$
\end{itemize}
for any $t \geq 0$ and $x \in D$.
Here $w_1(x)$ and $w_2(x)$ are positive definite functions, $W_1(x)$ and $W_2(x)$ are locally Lipschitz functions in $D$.
\end{definition}

\begin{definition}
\label{def3}
If condition (b) of Definition \ref{def1} satisfies $\dot{V} \leq \rho(x,t) W_1(x) < 0$ or $\dot{V} \geq \rho(x,t ) W_2(x)>0$,
then the system \eqref{eq1} is called strictly density.
\end{definition}

\begin{definition}
\label{def2}
If $\dot{V} \leq \rho(x,t) W_1(x) \leq 0$ in $D_S \times [0, +\infty)$,
then the density function $\rho(x,t)$ and the domain $D_S$ is called stable.
If $\dot{V} \geq \rho(x,t) W_2(x) > 0$ in $D_U \times [0, +\infty)$,
then the density function $\rho(x,t)$ and the domain $D_U$ is called unstable.
\end{definition}

\begin{proposition}
\label{Th1}
Let the system \eqref{eq1} be strictly density in $D_S$ and $D_U$.
If for each $t$ the condition $V(x_{s},t)-V(x_{u},t)>0$ is satisfied,
where $x_{s} \in D_S$ and $x_{u} \in D_U$, then the trajectories are attracted to the boundary between the regions $D_S$ and $D_U$.
If the system \eqref{eq1} for each $t$ satisfies the condition $V(x_{s},t)-V(x_{u},t)<0$,
then the trajectories of the system move away from the boundary between the regions $D_S$ and $D_U$.
\end{proposition}


\begin{proof}
Let the condition $V(x_{s},t)-V(x_{u},t)>0$ be satisfied for each $t \geq 0$, where $x_{s} \in D_S$ and $x_ {u} \in D_U$.
Since the system is strictly density, then, according to Definition \ref{def3}, the condition $\dot{V} \leq \rho(x,t) W_1(x) < 0$ is satisfied in $D_S $,
and the condition $\dot{V} \geq \rho(x,t) W_2(x) > 0$ is satisfied in $D_U$.
Hence, the boundary between $D_S$ and $D_U$ is a set to which the trajectories of the system are attracted.

Let the condition $V(x_{s},t)-V(x_{u},t)<0$ be satisfied for each $t \geq 0$, where $x_{s} \in D_S$ and $x_{u } \in D_U$.
According to Definition \ref{def3}, $\dot{V} \leq \rho(x,t) W_1(x) < 0$ is satisfied in $D_S$, and $\dot{V } \geq - \rho(x,t) W_2(x) > 0$.
This means that the separation boundary of $D_S$ and $D_U$ is the set that the trajectories of the system leave.
\end{proof}


\begin{remark}
\label{Rem1}
In Proposition \ref{Th1} and in its proof, under the attraction of trajectories to a set we may consider the cases when trajectories approaching a given set over time or finding trajectories in some neighborhood of a given set.
Moreover, the size of this neighborhood can remain the same or increase over time.
It depends on the density of space.
Let us consider some cases:
\begin{itemize}
\item if in the neighborhood of the boundary between the regions $D_S$ and $D_U$ the density value decreases to zero, then the trajectories of the system do not approach this boundary.
The trajectories can be located in the boundary vicinity or move away from boundary;
\item if in the neighborhood of the boundary between the regions $D_S$ and $D_U$ the value of density increases indefinitely, then the trajectories of the system will approach this boundary.
\end{itemize}
\end{remark}


To illustrate the conclusions in Proposition \ref{Th1} and Remark \ref{Rem1}, consider the following example.


\textit{Example 3.}
Consider again the system \eqref{eq_ex1_0}, where only $x \in \mathbb R_+$.

Choose the quadratic function $V=0.5x^2$.
Then $\dot{V} = - \rho(x,t) x^2$.
Hence $\dot{V}<0$ in $D_S=\{x,t \in \mathbb R_+: \rho(x,t)>0\}$ and $\dot{V}>0$ in the region $D_U=\{x,t \in \mathbb R_+: \rho(x,t)<0\}$.

According to Definition \ref{def3} the system \eqref{eq_ex1_0} is strictly dense for $x \in \mathbb R_+$.
Let us analyze Proposition \ref{Th1}.
We fix an arbitrary $t=t_1$.
Then $V(x_s(t_1))-V(x_u(t_1))>0$.
Obviously, this difference will be valid for any fixed $t$.
Hence, according to Proposition \ref{Th1}, the trajectories of the system will be attracted to the separation boundary of $D_S$ and $D_U$.

Let the density function be given as $\rho(x,t)=x-w$, $w(t)=e^t$.
In this case $\lim\limits_{t \to \infty} (w(t)-x(t)) = \textup{const}$ (see Fig.~\ref{Fig_Alex1}, \textit{a}).
If $w(t)=e^{e^t}$, then the difference between $w(t)$ and $x(t)$ increases with time (see Fig.~\ref{Fig_Alex1}, \textit{b}).
This is due to the fact that $w(t)$ is an unbounded function, and the dense value of $|\rho(x,t)|$ decreases as $x$ approaches $w$.
As a result, $x(t)$ tries to get closer to $w(t)$, but does not succeed because of the low density of space in the neighborhood of $w(t)$ and the high growth rate of $w(t)$.

\begin{figure}[h]
\begin{center}
\begin{minipage}[h]{0.8\linewidth}
\center{\includegraphics[width=0.7\linewidth]{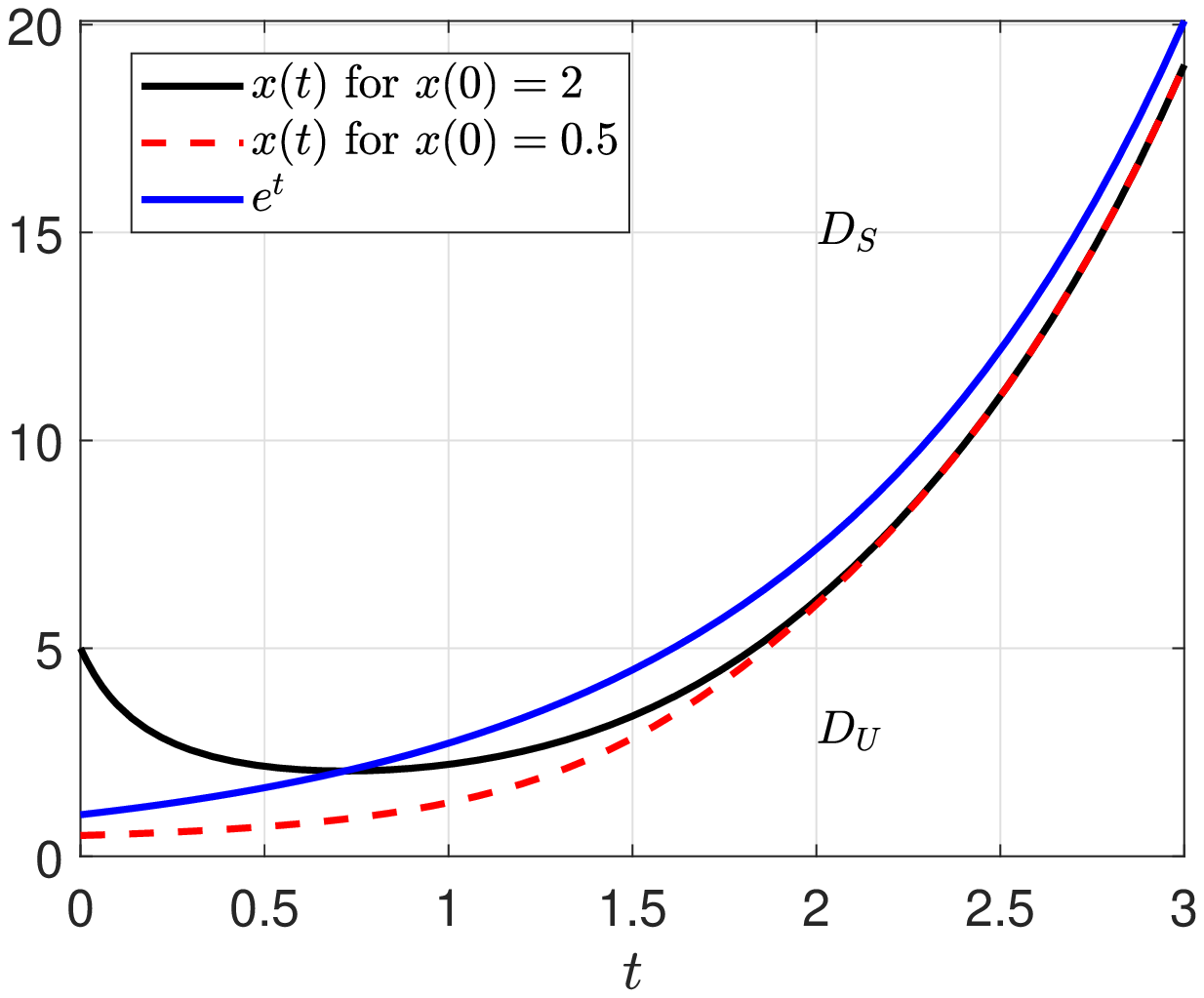}} \\ \textit{a}
\end{minipage}
\vfill
\begin{minipage}[h]{0.8\linewidth}
\center{\includegraphics[width=0.7\linewidth]{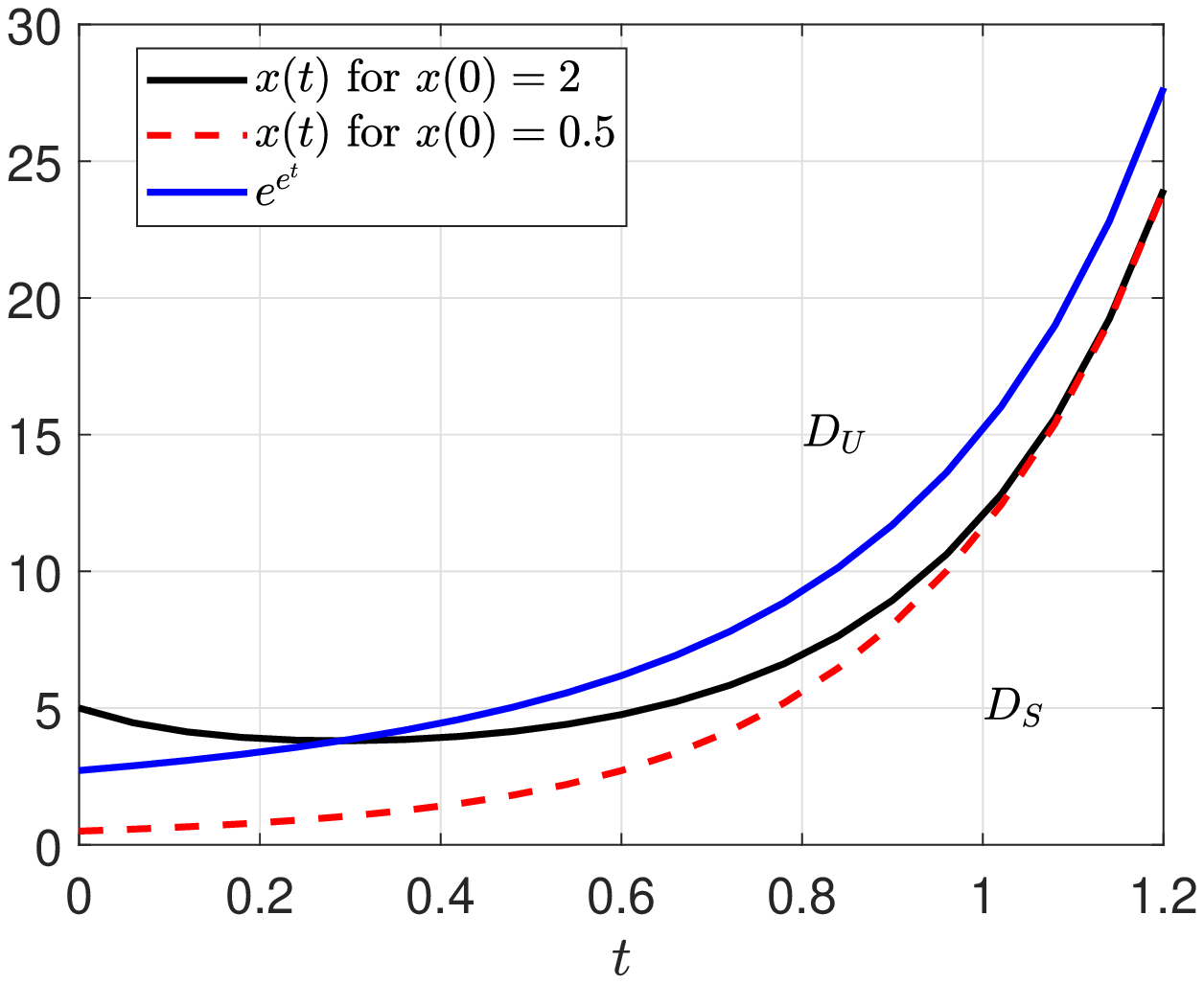}} \\ \textit{b}
\end{minipage}
\caption{Tracking $x(t)$ for unbounded signals $w(t)=e^t$ (\textit{a}) and $w(t)=e^{e^t}$ (\textit{b}).}
\label{Fig_Alex1}
\end{center}
\end{figure}

Let the space density be given as $\rho=w \textup{sign}(x-w)$, $w(t)=e^t$ or $w(t)=e^{e^t}$.
In this case, the space density in the neighborhood of $w(t)$ increases with $w(t)$, which guarantees that $x$ approaches $w$ as $t \to \infty$ (see Fig.~\ref{Fig_Alex2}).

\begin{figure}[h]
\begin{center}
\begin{minipage}[h]{0.8\linewidth}
\center{\includegraphics[width=0.7\linewidth]{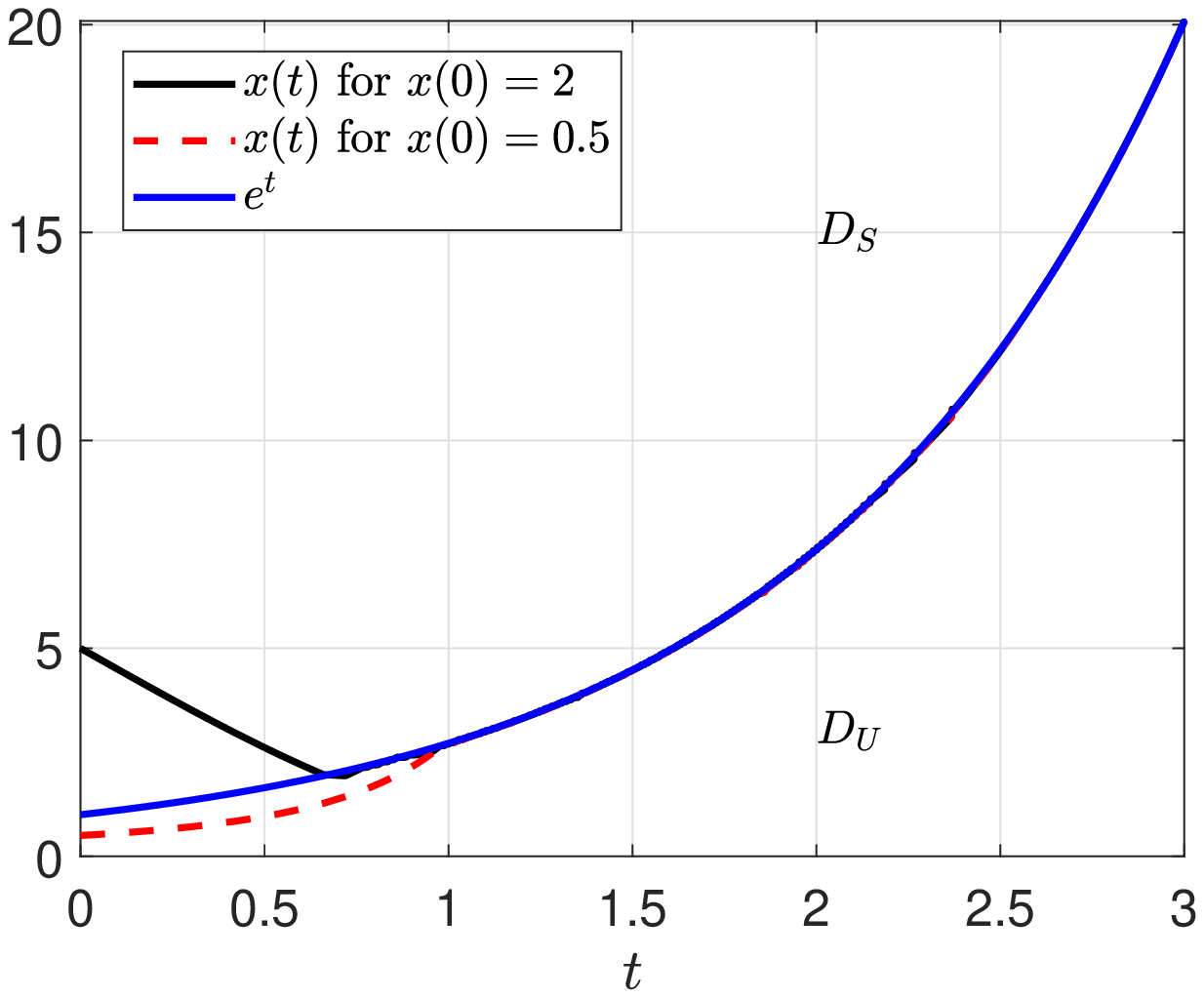}} \\ \textit{a}
\end{minipage}
\vfill
\begin{minipage}[h]{0.8\linewidth}
\center{\includegraphics[width=0.7\linewidth]{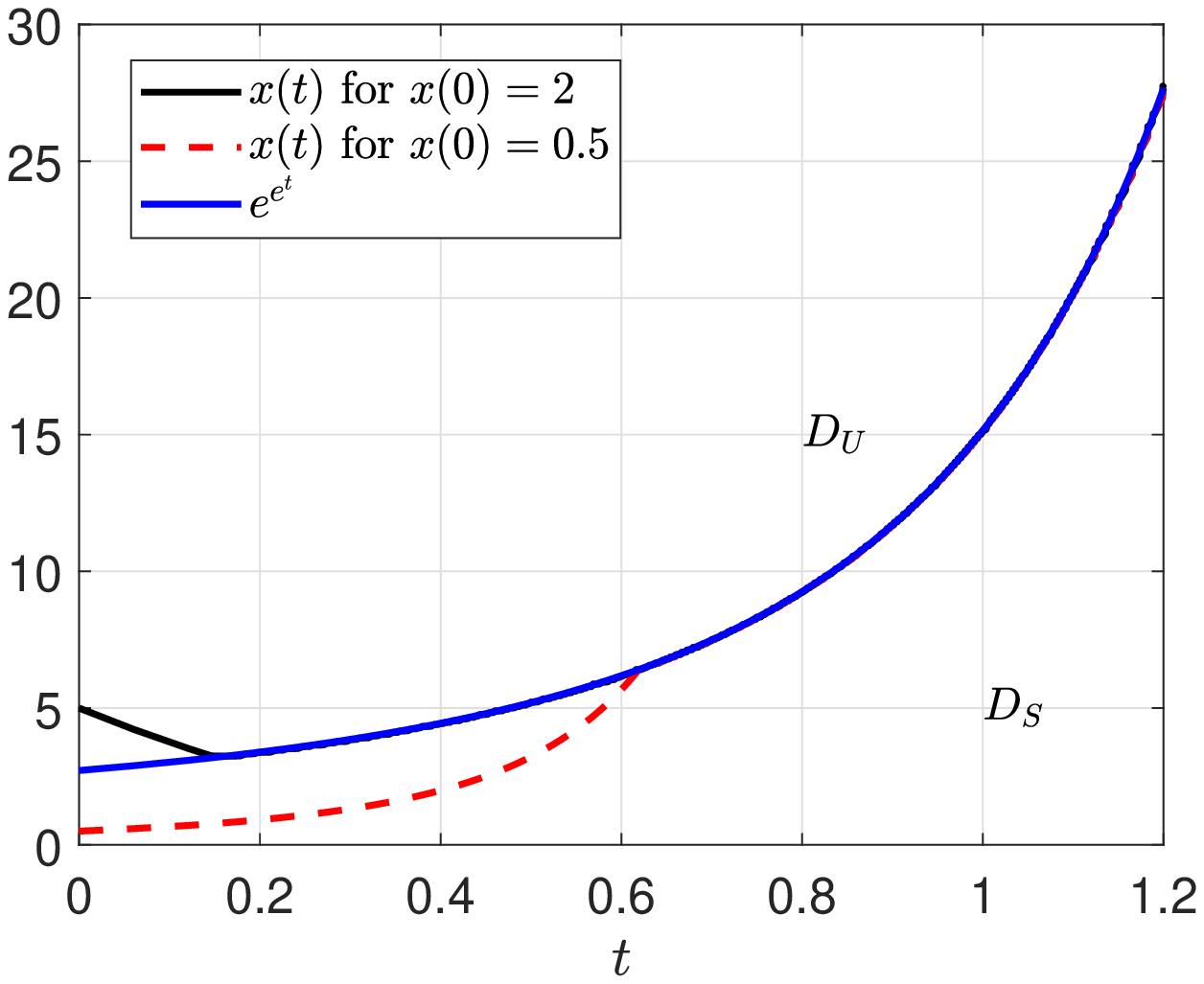}} \\ \textit{b}
\end{minipage}
\caption{Tracking $x(t)$ for unbounded signals $w(t)=e^t$ (\textit{a}) and $w(t)=e^{e^t}$ (\textit{b}).}
\label{Fig_Alex2}
\end{center}
\end{figure}


Also, when studying density systems, we will single out special areas.
We have already considered them earlier as gray areas in the figures. 
Now we will define them.

\begin{definition}
\label{def4}
If $\dot{V} \leq \rho(x,t) W_1(x) < 0$ in a neighborhood of $D_{bh} \times [0, +\infty)$,
there are no solutions in $D_{bh} \times [0, +\infty)$ \eqref{eq1}
and the density value increases to infinity approaching this area,
then the domain $D_{bh}$ is called absolutely stable.
\end{definition}

\begin{definition}
\label{def5}
If $\dot{V} \geq \rho(x,t) W_2(x) > 0$ in a neighborhood of $D_{wh} \times [0, +\infty)$,
there are no solutions in $D_{wh} \times [0, +\infty)$ \eqref{eq1}
and the density value increases to infinity approaching this area,
then the domain $D_{wh}$ is called absolutely unstable.
\end{definition}

\textit{Example 4.}
Consider the system \eqref{eq_ex2_1}, where $\rho_1(x,t)=\rho_2(x,t)=\rho(x)$.
Choose the quadratic function \eqref{eq_ex2_2}.
Then $\dot{V} = - \rho(x) (x_1^2+x_2^2)$.

If $\rho(x)=e^{(x_1^2+x_2^2-1)^{-0.98}}$, then all trajectories tend to $D_{bh}=\{x \in \mathbb R ^2: x_1^2+x_2^2 \leq 1\}$ (see Fig.~\ref{Fig_bh}, \textit{a}), that is the trajectories of the system will be attracted to the region $D_{bh}$ from any initial conditions. 
The density function value increases to infinity at its boundary.
If $\rho(x)=-\ln(x_1^2+x_2^2-1)$, then trajectories with initial conditions from $x_1^2+x_2^2 \geq \sqrt{2}$ remain in this region, 
but all trajectories with initial conditions from the region $1 \leq x_1^2+x_2^2 \leq \sqrt{2}$ cannot leave this region and will be attracted to the region $D_{bh}=\{x \in \mathbb R^2: x_1^2+x_2^2 \leq 1\}$ (see Fig.~\ref{Fig_bh}, \textit{b}), whose density increases to infinity at its boundary.

\begin{figure}[h]
\begin{center}
\begin{minipage}[h]{0.8\linewidth}
\center{\includegraphics[width=0.7\linewidth]{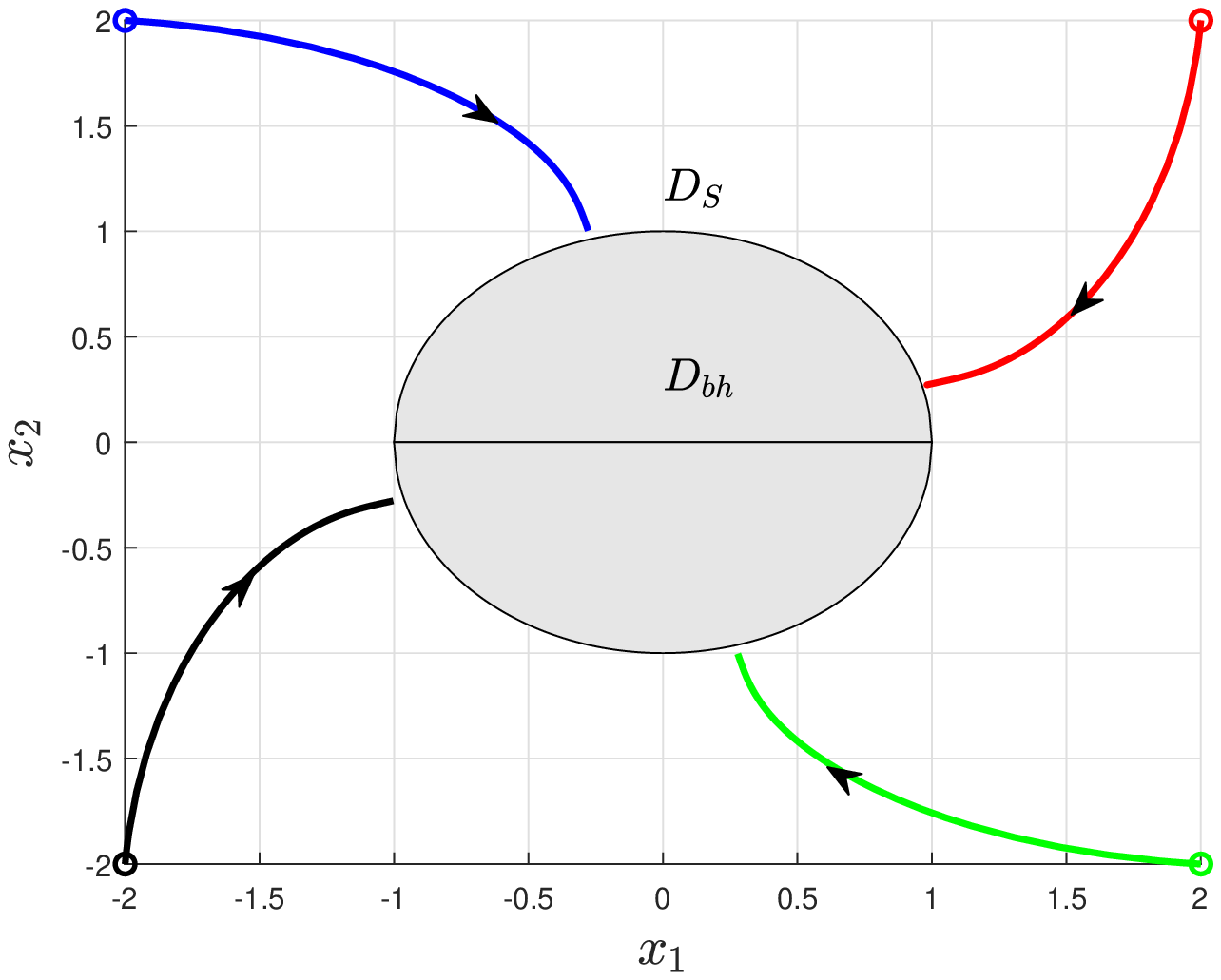}} \\ \textit{a}
\end{minipage}
\vfill
\begin{minipage}[h]{0.8\linewidth}
\center{\includegraphics[width=0.7\linewidth]{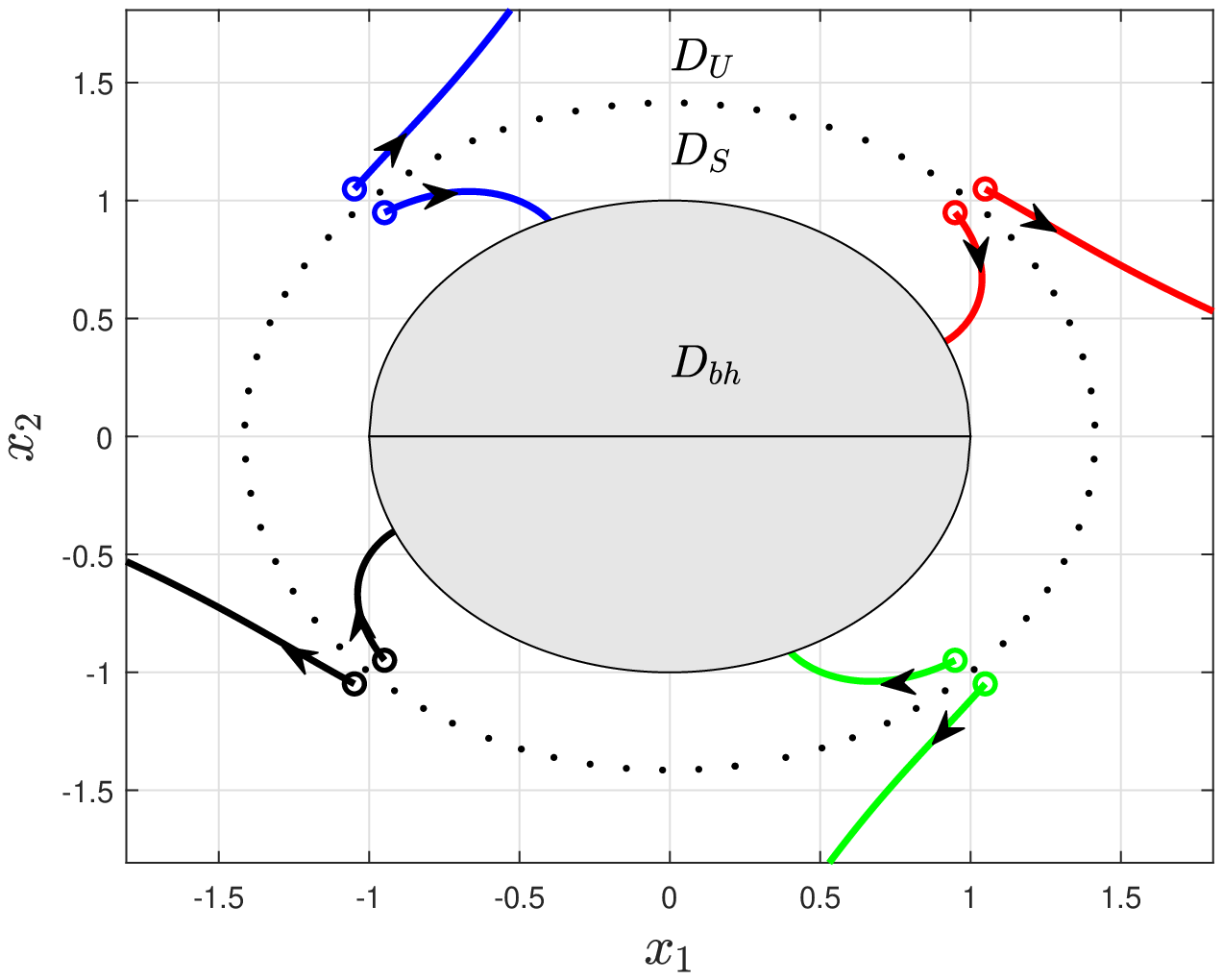}} \\ \textit{b}
\end{minipage}
\caption{The phase portrait of the system \eqref{eq_ex2_1} with density functions $\rho(x)=e^{(x_1^2+x_2^2-1)^{-0.98}}$ (\textit{a}) and $\rho(x)=-\ln(x_1^2+x_2^2-1)$ (\textit{b}).}
\label{Fig_bh}
\end{center}
\end{figure}

If $\rho(x)=-e^{(x_1^2+x_2^2-1)^{-0.98}}$, then all trajectories are moved away from the region $D_{wh}=\{x \in \mathbb R^2: x_1^2+x_2^2 \leq 1\}$ (see Fig.~\ref{Fig_wh}, \textit{a}) and will never be able to approach the border of this region, where the density is increased to infinity.
If $\rho(x)=\ln(x_1^2+x_2^2-1)$, then trajectories with initial conditions from $x_1^2+x_2^2 \geq \sqrt{2}$ remain in this region, but all trajectories are moved away from the region $D_{wh}=\{x \in \mathbb R^2: x_1^2 +x_2^2 \leq 1\}$ if initial conditions belong to $1 \leq x_1^2+x_2^2 \leq \sqrt{2}$ (see Fig.~\ref{Fig_wh}, \textit{b}).

\begin{figure}[h]
\begin{center}
\begin{minipage}[h]{0.8\linewidth}
\center{\includegraphics[width=0.7\linewidth]{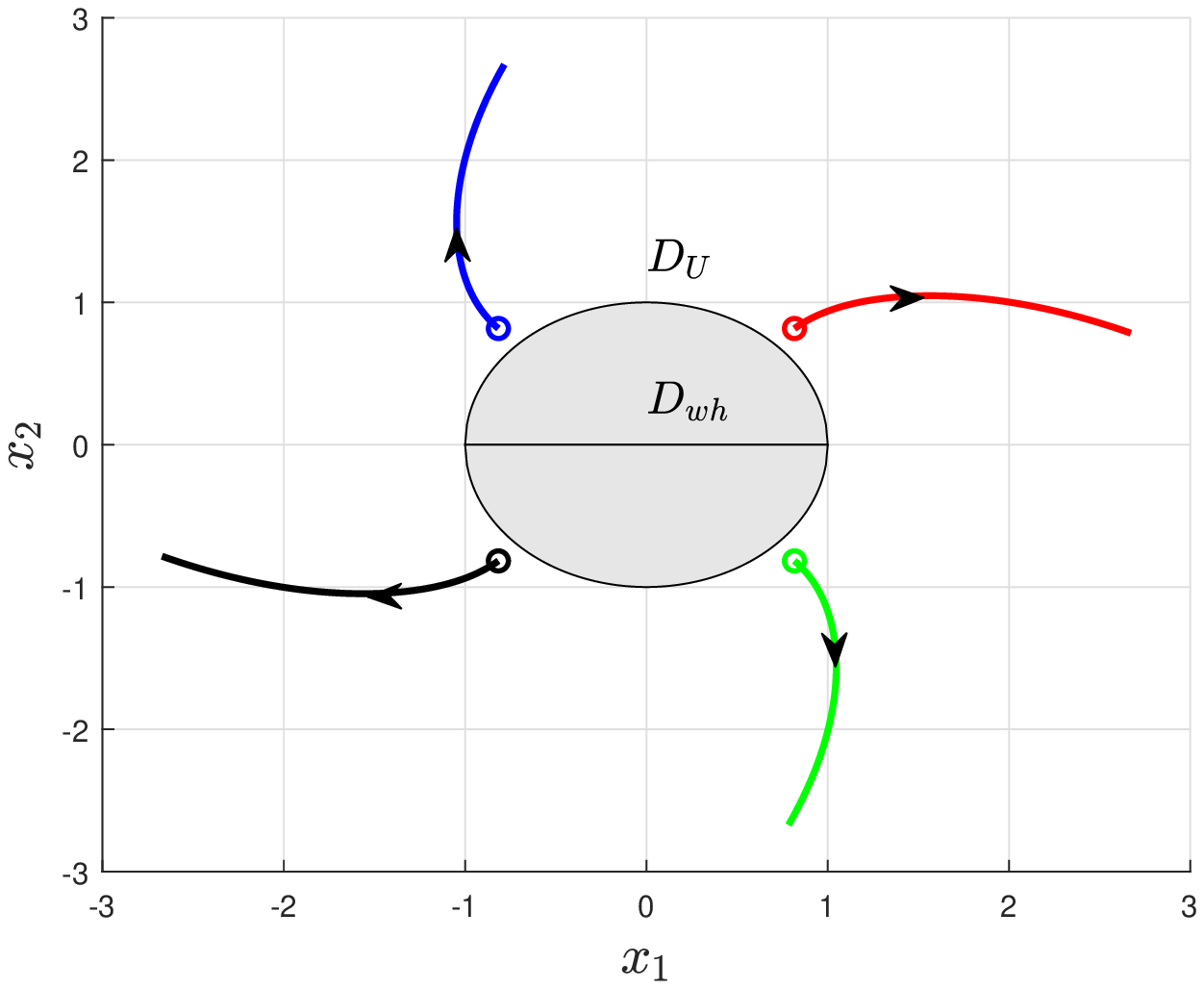}} \\ \textit{a}
\end{minipage}
\vfill
\begin{minipage}[h]{0.8\linewidth}
\center{\includegraphics[width=0.7\linewidth]{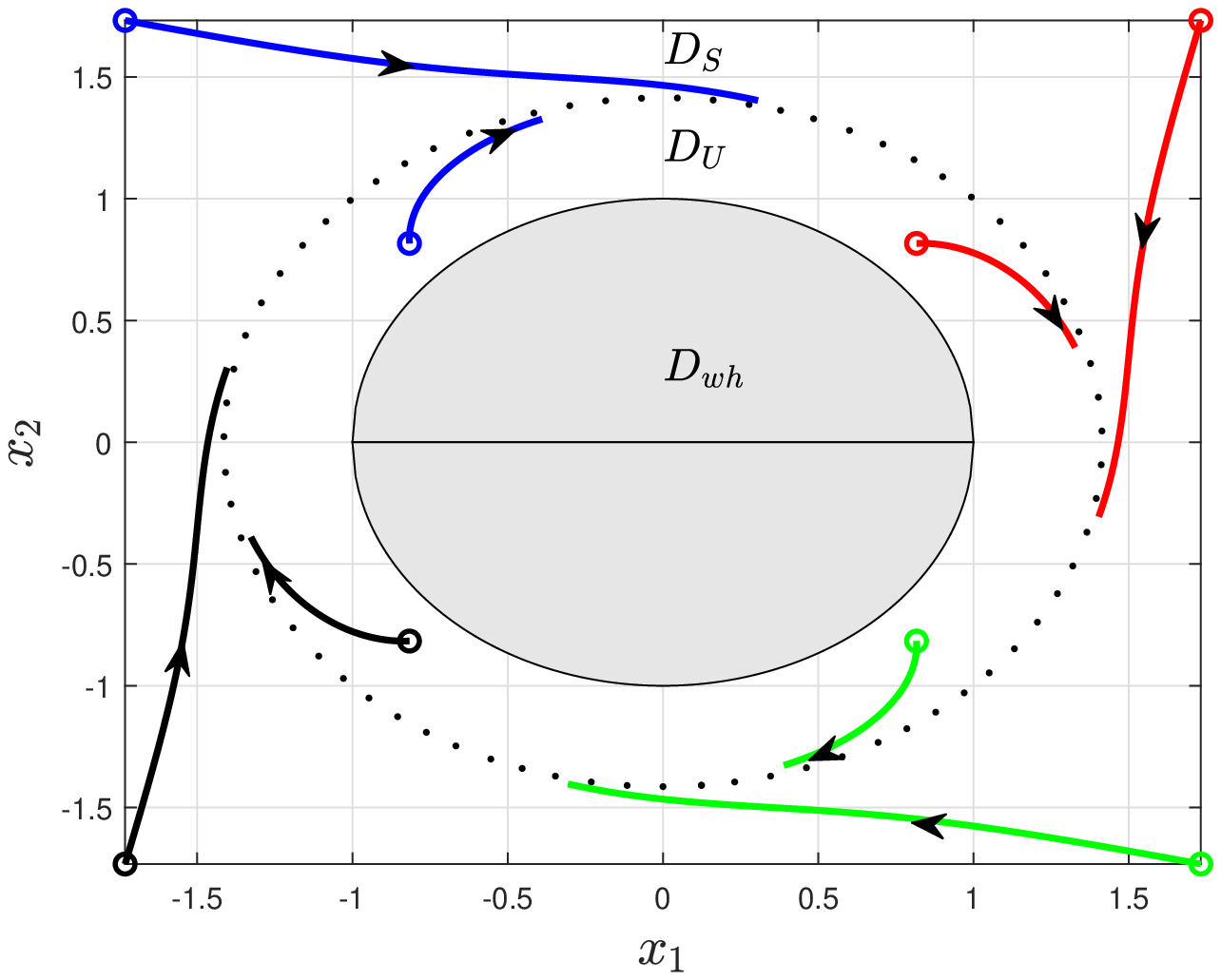}} \\ \textit{b}
\end{minipage}
\caption{The phase portrait of the system \eqref{eq_ex2_1} with density functions $\rho(x)=-e^{(x_1^2+x_2^2-1)^{-0.98}}$ (\textit{a}) and $\rho(x)=\ln(x_1^2+x_2^2-1)$ (\textit{b}).}
\label{Fig_wh}
\end{center}
\end{figure}


\begin{remark}
Let us explain the physical meaning of the density systems.
If the density function is explicitly presented on the right-hand side of the system equation, for example, in the form $\dot{x}=\rho(x)f(x)$, then the space density value directly affects the phase flow.
Thus, if $\rho(x)=1$, then we have an initial system of the form $\dot{x}=f(x)$.
If $\rho(x)>0$, then the presence of the density function does not qualitatively affect the equilibrium positions and their types, but quantitatively affects the phase portrait.
At $0<\rho(x)<1$ the value of the phase vector flow is decreased because the space density is decreased.
When $\rho(x)>1$ the value of the phase vector flow is increased due to the space density is increased.
When the sign of the density function changes, the phase portrait changes qualitatively.

Condition (b) of Definition \ref{def1} can be interpreted as the rate of change of the phase volume, given by the function $V(x,t)$ and taking into account the density of space.

Here are models of real systems:
\begin{itemize}

\item consider the pendulum equations $\dot{x}_1=x_2$, $\dot{x}_2=-\frac{g}{l} \sin x_1 - \frac{k}{m}x_2$, 
where $x_1 $ is an angle of pendulum deviation from the vertical axis, 
$x_2$ is the angular velocity of the pendulum, 
$g$ is the gravity acceleration, 
$l$ is the pendulum length, 
$k$ is the friction coefficient \cite{Khalil02}.
Choosing the quadratic function in the form of the total energy as $V=\frac{g}{l} (1-\cos x_1)+0.5x_2^2$, we obtain $\dot{V}=-\frac{k}{m} x_2^2$.
If we consider the density function $\rho(x,t)=k$, then in the absence of friction ($\rho(x,t)=0$) we have undamped oscillations, and damped oscillations in the presence of friction ($\rho(x,t ) \neq 0$);

\item in \cite{Arnold98} different types of breeding patterns can be written as $\dot{x}=\rho(x)x$, where $x$ is the size of the biological population.
For $\rho(x)=k>0$ we have a normal reproduction model, for $\rho(x)=kx$ we have an explosion model, for $\rho(x)=1-x$ we have a logistic curve model;

\item absolutely stable and absolutely unstable regions from Definitions \ref{def4} and \ref{def5} can be found as the simplest models of black and white holes, respectively, see \cite{Carroll04}.
\end{itemize}
\end{remark}


\section{\uppercase{Density control}}
\label{Sec3}

In this section, we consider several examples of design of the control laws in order to obtain the closed-loop systems that are described by density systems.

\subsection{Plants with known parameters}

Consider the system
\begin{equation}
\label{eq1_simple}
\begin{array}{l}
Q(p)y(t)=kR(p)u(t),
\end{array}
\end{equation}
where $y \in \mathbb R$ is the output,
$u \in \mathbb R$ is the control,
$Q(p)$ and $R(p)$ are linear differential operators with constant known coefficients,
$R(\lambda)$ is Hurwitz polynomial, 
$\lambda$ is a complex variable, 
$k>0$ is the high-frequency gain.

If the relative degree of \eqref{eq1_simple} is equal to $1$ (i.e., $\deg Q(p) - \deg R(p) = 1$), then the control law
\begin{equation}
\label{eq2_simple}
\begin{array}{l}
u(t)=-\frac{Q(p)}{kpR(p)}\rho(y,t)y(t)
\end{array}
\end{equation}
gets the system \eqref{eq1_simple} to the form
\begin{equation}
\label{eq3_simple}
\begin{array}{l}
\dot{y}(t)=-\rho(y,t)y(t),
\end{array}
\end{equation}
which has the structure of a density system.
In particular, examples of specifying the density function $\rho(y,t)$ are considered in Example 1.

If the relative degree of \eqref{eq1_simple} is greater than $1$ (i.e., $\deg Q(p) - \deg R(p) = \gamma > 1$), then the control law
\begin{equation}
\label{eq4_simple}
\begin{array}{l}
u(t)=-\frac{Q(p)}{kp R(p) (\mu p+1)^{\gamma-1}}\rho(y,t)y(t)
\end{array}
\end{equation}
where $\mu>0$ is a sufficiently small number, leads to the closed-loop system of the form
\begin{equation}
\label{eq5_simple}
\begin{array}{l}
\dot{y}(t)=\frac{1}{(\mu p+1)^{\gamma-1}}\rho(y,t)y(t).
\end{array}
\end{equation}
For $\mu = 0$ the system \eqref{eq5_simple} has the structure of a density system \eqref{eq3_simple}.
Let the density function $\rho(y,t)$ be chosen such that the solutions of the density system \eqref{eq5_simple} are asymptotically stable at $\mu = 0$. 
Then, according to \cite{Khalil02,Vasilyeva73}, there exists a sufficiently small $\overline{ \mu}>0$ such that for $\mu<\overline{\mu}$ the solutions of the system \eqref{eq5_simple} for $0<\mu<\overline{\mu}$ are sufficiently close to the solution of \eqref{eq5_simple} for $\mu=0$.


\subsection{Plants with unknown parameters}

Consider the system \eqref{eq1_simple} with unknown parameters of the operators $Q(p)$ and $R(p)$, as well as unknown value of $k>0$.
Let the relative degree of the system be equal to $1$.
All the obtained results can be extended to systems with a relative degree greater than $1$, for example, by using additional methods \cite{Fradkov99,Annaswamy21}.
In this paper, we consider only systems with a relative degree of $1$ in order to avoid cumbersome derivations on overcoming the problem of a high relative degree.

Let us rewrite the operators $Q(p)$ and $R(p)$ as $Q(p)=Q_m(p)+\Delta Q(p)$ and $R(p)=R_m(p)+\Delta R( p)$,
where $Q_m(\lambda)$ and $R_m(\lambda)$ are arbitrary Hurwitz polynomials of orders $n$ and $n-1$, respectively,
the orders of $\Delta Q(p)$ and $\Delta R(p)$ are $n-1$ and $n-2$, respectively.
Choosing $Q_m(\lambda)/R_m(\lambda)=\lambda$ and extracting the integer part in
$\frac{\Delta Q(\lambda)}{Q_m(\lambda)}=k_{0y}+\frac{\Delta \tilde{Q}(\lambda)}{R_m(\lambda)}$, rewrite \eqref{eq1_simple} as follows
\begin{equation}
\label{eq0_adapt}
\begin{array}{l}
\dot{y}(t)=k\left( u(t)+\frac{\Delta R(p)}{R_m(p)}u-\frac{\Delta \tilde{ Q}(p)}{R_m(p)}y-k_{0y}y \right).
\end{array}
\end{equation}

Let $c_0=col\{c_{0y},c_{0u},k_{0y}\}$ be the vector of unknown parameters,
$\Delta \tilde{Q}(p)=c_{0y}^{\rm T}[1~ p~ ... ~p^{n-2}]$ and $\Delta R(p)= c_{0u}^{\rm T}[1~ p~ ... ~p^{n-2}]$.
Also, consider the regression vector $w=col\{V_y,V_u,y\}$ and filters
\begin{equation}
\label{eq1_adapt_filt}
\begin{array}{l}
\dot{V}_y=FV_y+by,
\\
\dot{V}_u=FV_u+bu.
\end{array}
\end{equation}
Here $F \in \mathbb R^{(n-1) \times (n-1)}$ is Frobenius matrix with characteristic polynomial $R_m(\lambda)$,
$b=col\{0,...,0,1\}$.

Taking into account the introduced notation, the equation \eqref{eq0_adapt} can be rewritten in the form
\begin{equation}
\label{eq1_adapt}
\begin{array}{l}
\dot{y}(t)=k(u(t)-c_0)^{\rm T}w(t).
\end{array}
\end{equation}

Introduce the control law
\begin{equation}
\label{eq2_adapt}
\begin{array}{l}
u(t) = c^{\rm T}(t)w(t) + \rho(y,t).
\end{array}
\end{equation}

Substituting \eqref{eq2_adapt} into \eqref{eq1_adapt}, we get the closed-loop system
\begin{equation}
\label{eq3_adapt}
\begin{array}{l}
\dot{y}(t)=\rho(y,t)+k[c(t)-c_0^{\rm T}w(t)].
\end{array}
\end{equation}

\begin{theorem}
\label{th_adapt}
The control law \eqref{eq2_adapt} together with the adaptation algorithm
\begin{equation}
\label{eq6_adapt}
\begin{array}{l}
\dot{c} = -\gamma y w
\end{array}
\end{equation}
where $\gamma>0$, transforms the system \eqref{eq1_adapt} to the density type.
If for $t \to \infty$ we have a stable density function $\rho(y,t)$ with a ultimately stable set in a neighbourhood of zero,
then the condition $\lim\limits_{t \to \infty} y(t) = 0$ holds and all signals in the closed-loop system \eqref{eq3_adapt} with \eqref{eq1_simple}, \eqref{eq1_adapt_filt}, \eqref{eq2_adapt}, and \eqref{eq6_adapt}  are ultimately bounded.
\end{theorem}

\begin{remark}
The proposed control law \eqref{eq2_adapt} consists of the classical part $c^{\rm T}(t)w(t)$, the classical adaptation algorithm \eqref{eq6_adapt} (see, i.g. \cite{Fradkov99}), and a new component $\rho(y,t)$ in \eqref{eq2_adapt} that determines the density of space. 
However, as will be shown in the examples below, a new control law will allow achieving new control goals in comparison with \cite{Fradkov99,Annaswamy21}.
\end{remark}

\begin{proof}
To analyse the stability of the closed-loop system, choose Lyapunov function in the form
\begin{equation}
\label{eq4_adapt}
\begin{array}{l}
V=\frac{1}{2}y^2 + \frac{k}{2\gamma}(c(t)-c_0)^{\rm T}(c(t)-c_0).
\end{array}
\end{equation}

Let us find the full time derivative of \eqref{eq4_adapt} along the solutions \eqref{eq3_adapt}, \eqref{eq6_adapt} and rewrite the result as follows
\begin{equation}
\label{eq7_adapt}
\begin{array}{l}
\dot{V}=\rho(y,t) y.
\end{array}
\end{equation}
As a result, we got a density system.

If for $t \to \infty$ we have a stable density function $\rho(y,t)$ with a ultimately stable set in a neighbourhood of zero,
then $\rho(y,t)$ is chosen such that $\rho(y,t)y<0$.
Therefore, we have $\lim\limits_{t \to \infty} y(t) = 0$.
Expression \eqref{eq3_adapt} implies that $\lim\limits_{t \to \infty} (c(t)-c_0)^{\rm T}(t) w(t) = 0$.
The boundedness of $V_y(t)$ follows from the first equation \eqref{eq1_adapt_filt}, the boundedness of $y(t)$, and Hurwitz matrix $F$.
Putting \eqref{eq2_adapt} into the second equation \eqref{eq1_adapt_filt}, we get
\begin{equation}
\label{eq8_adapt_filt_proof}
\begin{array}{lll}
\dot{V}_u & =FV_u+bc_0^{\rm T}w+b(c-c_0)^{\rm T}w
+ b\rho (y,t) 
\\
&=(F+bc_{0u})V_u+b[c_{0y}^{\rm T}V_y+k_{0y}y
\\
&+(c-c_0)^{\rm T}w + \rho(y,t)].
\end{array}
\end{equation}
The matrix $F+bc_{0u}$ has Hurwitz characteristic polynomial $R(\lambda)$ due to the problem statement.
Hence, the function $V_u(t)$ is ultimately bounded because the term in square brackets in \eqref{eq8_adapt_filt_proof} is bounded.
Then the regression vector $w(t)$ is also ultimately bound.
From the condition $\lim\limits_{t \to \infty} y(t) = 0$ and ultimately boundedness of $w(t)$ it follows from \eqref{eq6_adapt} that $\lim\limits_{t \to \infty }\dot{c}(t)=0$.
Therefore, $c(t)$ is an ultimately bounded function.
Then \eqref{eq2_adapt} implies boundedness of the control law.
As a result, all signals are bounded in the closed-loop system.
\end{proof}

\textit{Example 3}.
Consider the unstable system \eqref{eq1_simple} with unknown 
$Q(p)=(p-1)^3$, $R(p)=(p+1)^2$, and $k=1$.

Define
$F=\begin{bmatrix}
0 & 1\\
-1 & -2
\end{bmatrix}$
in filters \eqref{eq1_adapt_filt}. 
Let $\gamma=0.1$ in \eqref{eq6_adapt}.
Consider various types of the density function $\rho(y,t)$ in \eqref{eq2_adapt}.

1) For $\rho(y,t)= - \alpha y$ the closed-loop system \eqref{eq3_adapt} has an equilibrium point $y=0$.
Substituting $\rho(y,t)$ into \eqref{eq7_adapt}, one gets $\dot{V}= -\alpha y^2 < 0$ in $D_S = \mathbb R \setminus \{0\}$.
We have obtained the classical problem of adaptive stabilization, which is described in detail in \cite{Fradkov99,Annaswamy21}.
Fig.~\ref{Fig_ad_1} (see only the trajectory entering the gray area) shows the transients for $\alpha=1$ and $p^2y(0)=py(0)=0$, $y(0)=4$.

2) For $\rho(y,t)= \alpha \ln\frac{g-y}{g+y}$, $g(t)>0$ the closed-loop system \eqref{eq3_adapt} has an equilibrium $y=0$.
Substituting $\rho(y,t)$ into \eqref{eq7_adapt}, we have $\dot{V}= \alpha \ln\frac{g-y}{g+y} y <0$ in $D_S = \{ y \in \mathbb R: -g<y<g\}$.
Moreover, $\rho(y,t) \to -\infty$ for $y \to g$ and $\rho(y,t) \to +\infty$ for $y \to -g$.
We have obtained a stabilization problem with symmetric constraints $-g$ and $g$.
Fig.~\ref{Fig_ad_1} shows the transients for $\alpha=1$ (trajectory inside the dotted tube), $p^2y(0)=py(0)=0$, $y(0)=4 $ and $g(t)=(4e^{-3t}+1)h(t)$,
$h(t)=\begin{cases}
    1 & t \leq 1,\\
    0.4 & t > 1.
  \end{cases}$
It can be seen that, in contrast to the classical control scheme \cite{Fradkov99,Annaswamy21} (the trajectory corresponding to $\rho(y,t)= - \alpha y$),
setting a density function of the form $\rho(y,t)= \alpha \ln\frac{g-y}{g+y}$ guarantees that the transients are in the tube at any time.

\begin{figure}[h]
\center{\includegraphics[width=0.5\linewidth]{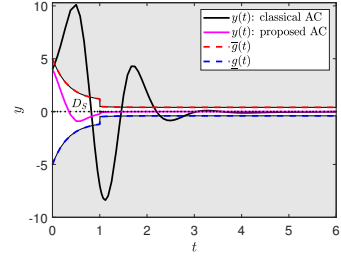}}
\caption{The transients in an adaptive control scheme with density functions $\rho(y,t)= - \alpha y$ (curve crossing the gray area) and $\rho(y,t)= \alpha \ln\frac{ g-y}{g+y}$ (a curve inside a tube with dashed borders).}
\label{Fig_ad_1}
\end{figure}

3) For $\rho(y,t)= \alpha \ln\frac{\overline{g}-y}{y-\underline{g}}$ the closed-loop system \eqref{eq3_adapt} has an equilibrium $y =\frac{\overline{g}+\underline{g}}{2}$.
Substituting $\rho(y,t)$ into \eqref{eq7_adapt}, we have $\dot{V}= \alpha \ln\frac{\overline{g}-y}{y-\underline{g}} y <0$ in
$D_S=\left\{y \in \mathbb R_+: \frac{\overline{g}+\underline{g}}{2} < y < \overline{g} \right\}$ and
$\dot{V}= \alpha \ln\frac{\overline{g}-y}{y-\underline{g}} y >0$ in
$D_U=\left\{y \in \mathbb R_+: \underline{g} < y < \frac{\overline{g}+\underline{g}}{2} \right\}$.
Also, $\dot{V}= \alpha \ln\frac{\overline{g}-y}{y-\underline{g}} y >0$ in
$D_U=\left\{y \in \mathbb R_-: \frac{\overline{g}+\underline{g}}{2} < y < \overline{g} \right\}$ and
$\dot{V}= \alpha \ln\frac{\overline{g}-y}{y-\underline{g}} y >0$ in
$D_S=\left\{y \in \mathbb R_-: \overline{g} < y < \frac{\overline{g}+\underline{g}}{2} \right\}$.
Moreover, $\rho(y,t) \to -\infty$ for $y \to \overline{g}$ and $\rho(y,t) \to +\infty$ for $y \to \underline{g }$ for $y \in \mathbb R_+$,
as well as $\rho(y,t) \to +\infty$ for $y \to \overline{g}$ and $\rho(y,t) \to -\infty$ for $y \to \underline{ g}$ for $y \in \mathbb R_-$.
We have obtained a stabilization problem with asymmetric constraints $\overline{g}$ and $\underline{g}$.
Fig.~\ref{Fig_ad_3} shows the transients for $\alpha=5$, $\overline{g}=4e^{-0.1t}+0.1$, $\underline{g}=3e^{-0.1 t}-0.1$ and $p^2y(0)=py(0)=0$, $y(0)=4$.

\begin{figure}[h]
\center{\includegraphics[width=0.5\linewidth]{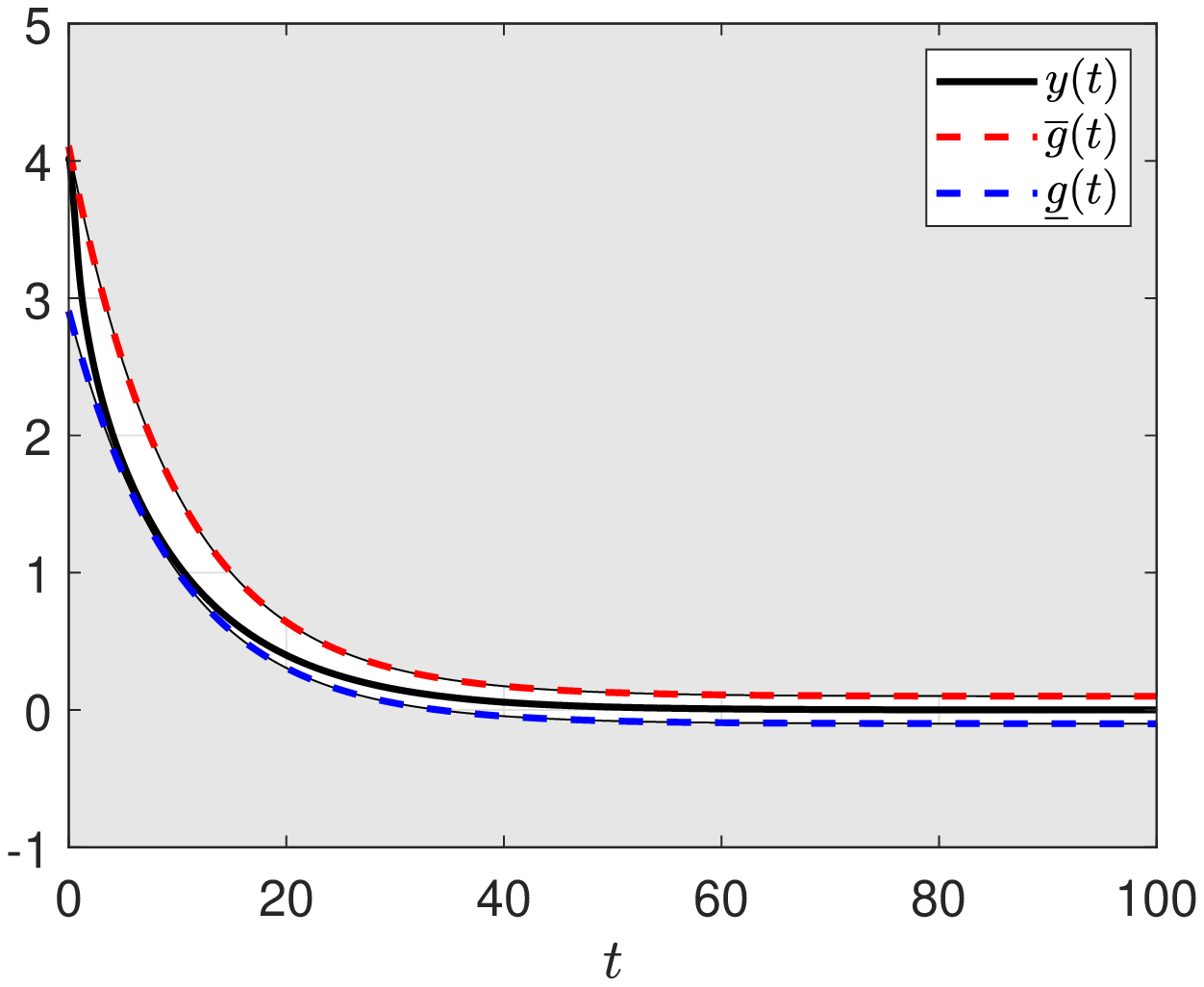}}
\caption{The transients in an adaptive control scheme with the density function $\rho(y,t)= \alpha \ln\frac{\overline{g}-y}{y-\underline{g}}$.}
\label{Fig_ad_3}
\end{figure}

4) For $\rho(y,t)= - \alpha (y-y_m)$ the closed-loop system \eqref{eq3_adapt} has an equilibrium $y=y_m$.
Substituting $\rho(y,t)$ into \eqref{eq7_adapt}, we have $\dot{V}= - \alpha (y-y_m) y <0$ in $D_S=\{y \in \mathbb R_+ : y >y_m \}$
and $\dot{V}= - \alpha (y-y_m) y >0$ in $D_U=\{y \in \mathbb R_+: y <y_m\}$.
Also, $\dot{V}= - \alpha (y-y_m) y <0$ in $D_S=\{y \in \mathbb R_-: y < y_m\}$ and $\dot{V}= - \alpha (y-y_m) y >0$ in $D_U=\{y \in \mathbb R_-: y > y_m\}$.
We get the problem of tracking $y$ to $y_m$.
Fig.~\ref{Fig_ad_4} shows the transients for $\alpha=100$, $y_m=e^{-0.1t}\sin(t)P(t)$, $P(t) \in [- 1,1]$ is a rectangular pulse generator with a switching period of $2.5$ [s] and $p^2y(0)=py(0)=0$, $y(0)=1$.

\begin{figure}[h]
\center{\includegraphics[width=0.5\linewidth]{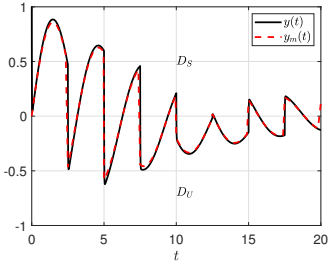}}
\caption{The transients in an adaptive control scheme with the density function $\rho(y,t)= - \alpha (y-y_m)$.}
\label{Fig_ad_4}
\end{figure}

5) For $\rho(y,t)= - \alpha \ln(y-g)$, $g(t) \geq -1$ the closed-loop system \eqref{eq3_adapt} has an equilibrium $y=g+1$.
Substituting $\rho(y,t)$ into \eqref{eq7_adapt}, we have $\dot{V}= - \alpha \ln(y-g) y <0$ in $D_S=\{y \in \mathbb R_+ : y > g+1\}$ and $\dot{V}= \alpha \ln(y-g) y >0$ in $D_U=\{y \in \mathbb R_+: y > g+1\}$.
Also, $\rho(y,t) \to -\infty$ for $y \to g$.
Therefore, one obtains the problem of sliding along the surface with border $g$.
Fig.~\ref{Fig_ad_5} shows the transients for $\alpha=10$, $g=2e^{-0.1t}-1$ and $p^2y(0)=py(0)=0$, $y(0)=4$.

\begin{figure}[h]
\center{\includegraphics[width=0.5\linewidth]{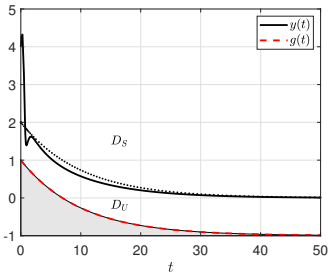}}
\caption{The transients in an adaptive control scheme with the density function $\rho(y,t)= - \alpha \ln(y-g)$.}
\label{Fig_ad_5}
\end{figure}

\section{\uppercase{Conclusion}}
\label{Sec6}

The paper considers a class of dynamical systems, called density systems, which contain the density function that specifies the properties of space.
By defining the properties of this function, one can influence the behaviour of the investigated system.
This conclusion is further used for the design of control laws.
It is shown that for various typos of the density function, it is possible to obtain both classical control laws and new ones that allow the formation of new target requirements for the system.
In particular, an example of design an adaptive control law with a guarantee of transients in a tube specified by the designer is given, while classical adaptive control provides only the ultimate boundedness of trajectories.
In this case, the parameters of the tube are set using the density function, which sets the density of the space.
The simulation results confirmed the theoretical conclusions.

In the paper, as an example of the application of the density function with known control schemes, it is shown how existing control algorithms can be modified to obtain a new quality of transients.
In the future works, the properties of density systems can be applied to more complex control algorithms, such as output control of systems with any relative degree, observer based control, sliding mode control, etc.

\end{document}